\documentclass[12pt]{article}
\usepackage{amsmath,amssymb,amsfonts,amsthm,graphicx}
\usepackage{tkz-graph}
\usepackage{color}
\usepackage{cite}
\usepackage{epsfig}
\usepackage{epstopdf}
\usepackage{footnote}
\usepackage{cite}
\usepackage{caption}
\usepackage{subcaption}
\usepackage{enumerate}
\usepackage{algorithm}
\usepackage{algpseudocode}
\usepackage{times}
\usepackage[top=3cm, bottom=3cm, left=2cm, right=2cm]{geometry}

\newtheorem{theorem}{Theorem}[section]
\newtheorem{lemma}[theorem]{Lemma}

\theoremstyle{definition}
\newtheorem{definition}[theorem]{Definition}

\theoremstyle{remark}
\newtheorem{remark}[theorem]{\bf Remark}

\algnewcommand\algorithmicswitch{\textbf{switch}}
\algnewcommand\algorithmiccase{\textbf{case}}
\algnewcommand\algorithmicassert{\texttt{assert}}
\algnewcommand\Assert[1]{\State \algorithmicassert(#1)}%
\algdef{SE}[SWITCH]{Switch}{EndSwitch}[1]{\algorithmicswitch\ #1\ \algorithmicdo}{\algorithmicend\ \algorithmicswitch}%
\algdef{SE}[CASE]{Case}{EndCase}[1]{\algorithmiccase\ #1}{\algorithmicend\ \algorithmiccase}%
\algtext*{EndSwitch}%
\algtext*{EndCase}%
\renewcommand\labelitemii{$\bullet$}
\title{An Efficient Algorithm for Mixed Domination on Generalized Series-Parallel Graphs}

\author{\small\it  M.Rajaati$^1$, P. Sharifani$^{2}$, A. Shakiba$^{3}$, M. R. Hooshmandasl$^{4}$, M. J. Dinneen$^{5}$   \\
	\footnotesize{$^{1,2,4}$Department of Computer Science, Yazd University, Yazd, Iran. }  \\
	\footnotesize{$^{3}$ Department of Computer Science, Vali-e-Asr University of Rafsanjan, Rafsanjan, Iran.}\\
	\footnotesize{$^{5}$Department of Computer Science, The University of Auckland, New Zealand.}\\
	\footnotesize{e-mail:$^1$m.rajaati@stu.yazd.ac.ir, $^2$ pouyeh.sharifani@gmail.com $^3$ali.shakiba@vru.ac.ir.}\\
 \footnotesize{$^4$hooshmandasl@yazd.ac.ir, $^5$m.dinneen@auckland.ac.nz.}
}

\date{}

\begin{document}

	\maketitle
	
	\begin{abstract}
		A mixed dominating set $S$ of a graph $G=(V,E)$ is a subset $ S \subseteq V \cup E$ such that each element $v\in (V \cup E) \setminus S$ is adjacent or incident to  at least one element in $S$.
		The mixed domination number $\gamma_m(G)$ of a graph $G$ is the minimum cardinality among all mixed dominating sets in $G$.
		The problem of finding $\gamma_{m}(G)$ is know to be  NP-complete. In
this paper, we present an explicit polynomial-time algorithm to construct a mixed
dominating set of size $\gamma_{m}(G)$ by a parse tree when $G$ is a generalized series-parallel graph.


	\end{abstract}

\section{Introduction}
A subset $S \subseteq V \cup E$ in graph $G=(V,E)$ is a mixed dominating set if for every $v\in (V \cup E) \setminus S$, where $v$ is either adjacent or incident to at least one element of $S$.
The mixed domination problem, also know as the total cover problem, is a variant of
classical dominating set problem and was introduced by Alavi et al.~in 1977~\cite{alavi1977total}.  
One of the known applications of the mixed domination problem is placing phase
measurement units (PMUs) in an electric power system~\cite{zhao2011algorithmic}. 
The minimum cardinality among all mixed dominating sets in $G$ is denoted by $\gamma_m(G)$.   
In~\cite{alavi1977total}, Alavi et al.~showed that this number is bounded from above by $\lceil n/2 \rceil$ for a connected graph $G$ of order $n$. They illustrated some extremal cases and gave some properties for connected graphs which have a total covering number equal to $\lceil n/2 \rceil$ in~\cite{alavi1992total}.
In~\cite{majumdar1992neighborhood}, Majumdar showed that the problem of finding $\gamma_m(G)$  is NP-complete  for general  graphs.
 Also, it is showed that even when the problem is restricted to chordal
graphs~\cite{hedetniemi1995domination}, planar bipartite graphs~\cite{manlove1999algorithmic},  and split graphs\cite{lan2013mixed, zhao2011algorithmic} is remain $NP$-complete. 
Finding a mixed dominating set whit minimum cardinality is tractable for some family of
graphs such as trees~\cite{adhar1994mixed,zhao2011algorithmic,lan2013mixed}, cactus
graphs~\cite{lan2013mixed} and graphs with bounded tree-width~\cite{rajaati2016fixed}.    

Rajaati et al.~presented a dynamic programming algorithm to solve the mixed domination problem on graphs with bounded tree-width by using tree decomposition of graphs but we use the parse tree of graphs in order to present an explicit polynomial-time algorithm to construct a mixed dominating set for generalized  series-parallel graphs.

In this paper, we propose a new dynamic programming algorithm to compute $\gamma_m(G)$ for a given generalized series-parallel graph $G$ in linear time. Moreover, we enumerate the number of $\gamma_m$-sets of $G$.
The rest of the paper is organized as follows. 
 In Section~\ref{Sec2}, we review basic definitions and  notions. 
  In Section~\ref{Sec3}, using a parse tree of the generalized series-parallel graph $G$, we present
a linear time algorithm to find a $\gamma_m$-set and determine the number of $\gamma_m$-sets for $G$. This section concludes by  analyzing the correctness and computational complexity of the proposed algorithms.

\section{Preliminaries}\label{Sec2}
In this section, we review some required graph theory and set our notation. For
notation and terminology that do not appear here, or for more details, an interested reader is advised to consult~\cite{west2001introduction}.
 
 In a graph $G = (V, E)$, the neighborhood of a vertex $v \in V$ is the set of all vertices adjacent to $v$ and is denoted by $N_G(v)$. The closed neighborhood of a vertex $v$ is defined as $N_G[v] = N_G(v) \cup \{v \}$.
Similarly, for an element $r\in V \cup E$, the mixed neighborhood of $r$ in $G$ is denoted by $N^{md}_{G}(r)$ and is defined $N^{md}_G(r) = \{ s\in V \cup E \mid   s\, \text{is adjacent or incident to}\, r \}$. Similarly, the closed mixed neighborhood of $r$ is denoted by $N^{md}_G[r]$ and equals $N^{md}_G[r] =N^{md}_G(r)\cup \{r\}$.

Domination in graphs and its variations are well studied topics in the literature~\cite{haynes1998fundamentals,haynes1997domination}. One of these variants is the mixed domination problem. 
A subset $ S \subseteq V \cup E$ is a mixed dominating set, if for every $r\in V \cup E$, it is the case that $|N^{md}_G[r] \cap S|\geq 1 $. The minimum cardinality of  such sets is denoted by $\gamma_{m}(G)$ and a $\gamma_m$-set for $G$ is a mixed dominating set of size $\gamma_{m}(G)$.

\begin{definition}[Generalized Series-Parallel Graphs~\cite{chebolu2012exact}]
	A generalized series-parallel, or GSP for short, is a graph $G=(V,E,s,t)$ with two distinguished vertices  $s,t\in V$ called terminals and is defined recursively as follows:
	\begin{itemize}
		\item[$(1)\;\;o_i$:]
		A graph $G$ consisting of two vertices connected by a single edge is a GSP.
		\item[$(2)\;\;o_s$:]
		Given two GSP graphs $G_1=(V_1,E_1,s_1,t_1)$ and $G_2=(V_2,E_2,s_2,t_2)$, the series operation of $G_1$ and $G_2$  is a new GSP graph $G=(V,E, s_1,t_2)$ denoted by $G_1 o_s G_2$ where
		$$V=V_1\cup V_2\setminus \{s_2\},$$
		and
		$$E=E_1 \cup E_2 \cup\{\{t_1,v\}: v\in N_{G_2}(s_2)\}\setminus \{\{s_2,v\}: v\in N_{G_2}(s_2)\}.$$
			\item[$(3)\;\;o_p$:]
		Given two GSP graphs $G_1=(V_1,E_1,s_1,t_1)$ and $G_2=(V_2,E_2,s_2,t_2)$, the parallel operation of $G_1$ and $G_2$  is  a new GSP graph  $G_1 o_p G_2=(V,E, s_1,t_1)$, where
				$$V=V_1\cup V_2\setminus \{s_2,t_2\},$$ and
		\begin{align*}
		E=\left(E_1 \cup E_2 \cup\{\{s_1,v\}: v\in N_{G_2}(s_2)\}\cup\{\{t_1,v\}: v\in N_{G_2}(t_2)\}\right)\setminus\\
		\left(\{\{s_2, v\}: v\in N_{G_2}(s_2)\} \cup \{\{t_2,v\}: v\in N_{G_2}(t_2)\}\right)
		\end{align*}
		
		\item[$(4)\;\;o_g$:]	Given two GSP graphs $G_1=(V_1,E_1,s_1,t_1)$ and $G_2=(V_2,E_2,s_2,t_2)$, the generalized series operation of $G_1$ and $G_2$  is a new GSP graph $G_1 o_g G_2=(V,E, s_1,t_1)$ where
			$$V=V_1\cup V_2\setminus \{s_2\}, $$ and
		$$E=E_1 \cup E_2 \cup\{\{t_1,v\}: v\in N_{G_2}(s_2)\}\setminus \{\{s_2,v\}: v\in N_{G_2}(s_2)\}.$$
		\item[$(5)\;\;\;\;\;\;$]		Any GSP graph is obtained by finite application of rules $(2)$  through $(4)$ by starting at rule $(1)$.
	\end{itemize}

\end{definition}
If we forbid rule $(4)$, then we obtain a subclass of GPSs called series-parallel or SP graphs. The application of these rules is illustrated in Figure~\ref{fig-gspop}. Note that in Figure~\ref{fig-gspop}, the graph $(\hat{G}o_g(G_1o_sG_2))o_s((G_1o_sG_2)o_p\hat{G})$ is a GSP; however it is not an $SP$.
\begin{figure}[!h]
	\centering
	\includegraphics[width=12cm]{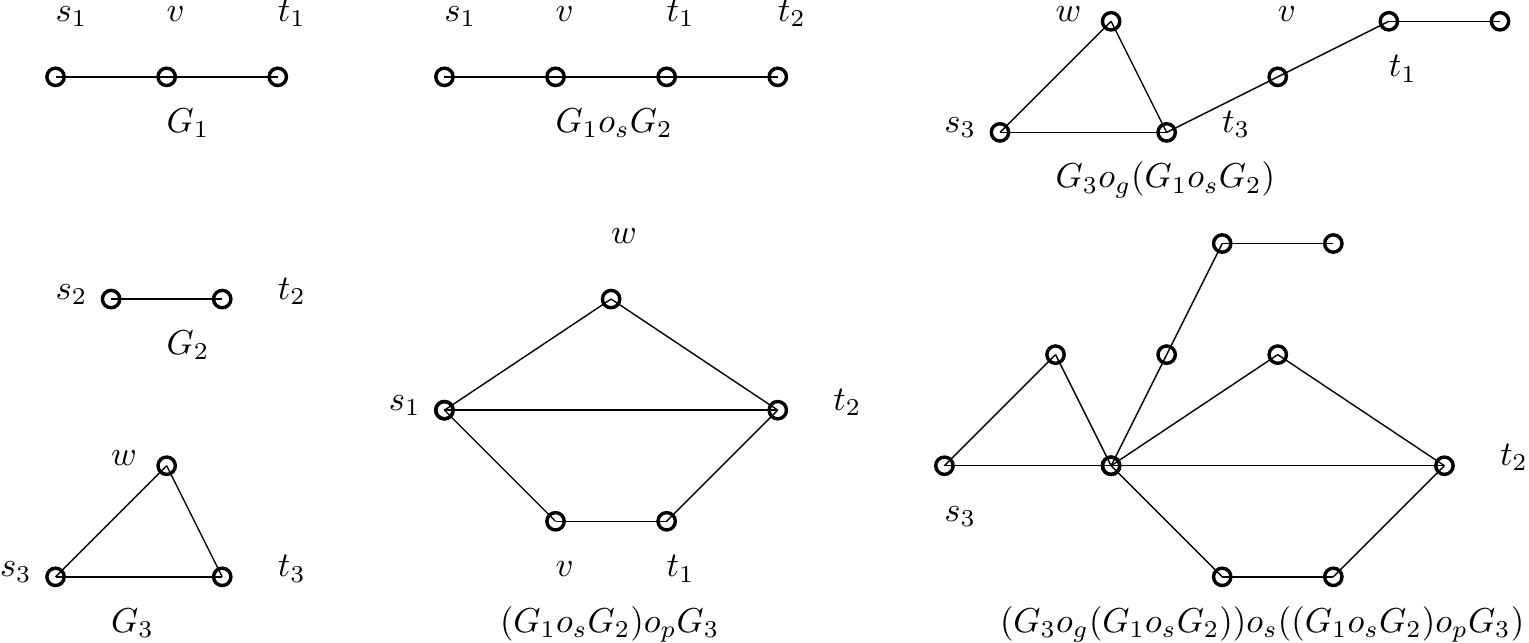}
	\caption{An illustration of constructing GSP graphs.}\label{fig-gspop}
\end{figure}
The concept of $p$-graph is defined for $SP$s in~\cite{kikuno1983linear} and  we generalize it to GSP graphs.   
\begin{definition}[$p$-graph]
		Let $G = (V, E,x,y)$  be a GSP and  $\hat{G} =(\hat{V}, \hat{E},\hat{x},\hat{y})$ be subgraph of $G$ satisfying  the following conditions:
	\begin{enumerate}
		\item
		$\hat{x}=x$ or there exists an edge $\{u,v\}  \in E \setminus \hat{E}$ such that $x \notin \hat{V}$  and $v = \hat{x}\in \hat{V}$.
		\item
		$\hat{y}=y$ or there exists an edge $\{w,z\}\in E \setminus \hat{E}$ such that $w \notin \hat{V}$  and $z = \hat{y}\in \hat{V}$.
	\end{enumerate}
	Then, $\hat{G}$ is called a $p$-graph of $G$.
\end{definition}

A generalized series-parallel graph $G$ can be represented by a binary parse tree $T$ which is defined as follows.

\begin{definition}[Binary Parse Tree for GSP Graphs~\cite{kikuno1983linear}]
	A binary parse tree for GSP $G$  is defined recursively as follows:
	\begin{enumerate}
		\item
		A tree consisted of a single vertex labeled $(u,v)_i$ is a binary parse tree for primitive  GSP  $G=(\{u,v\},\{u,v\}, u,v)$.
		\item
		Let $G=(V,E)$ be a GSP by some composition of two other GSPs $G_1$ and $G_2$, and  $T_1$ and $T_2$ be their binary parse trees, respectively. Then, a binary parse tree for $G$ is a tree with the root $r$ labeled as either  $(u,v)_s$, $(u,v)_p$ or $(u,v)_g$ depending on which operation is applied to generate $G$. Vertices $u$ and $v$ are terminals of $G$ and roots of $T_1$ and $T_2$ are the left and the right children of $r$, respectively.
		
		It is obvious that in any binary parse tree for a GSP $G$, every internal vertex of the tree  has exactly two children and there are $\vert E \vert$  leaves.
		\begin{remark}
		Note that when we use a label $(x,y)$, we do not care about the label being either $(x,y)_i,(x,y)_s,(x,y)_p$ and $(x,y)_g$.
		\end{remark}
		
		With $t$ be an internal vertex of a binary parse tree $T$ and for the
GSP $G$, let $\tau(t)$ denote the subtree of $T$ rooted at $t$.  Also, the left and the right subtree of $t$  are denoted by  $\tau_l(t)$  and  $\tau_r(t)$,  respectively. Then, the vertices of $T$  are labeled as follows:
		
		\begin{enumerate}
			\item
			For each edge $e=\{x,y\} \in E$, there exists exactly one leaf which is labeled by $(x,y)$ in $T$.
			\item
			For each internal vertex $t \in V_T$  that is  labeled by $(x,y)_s$, the root of $\tau_l(t)$  is labeled by $(x,z)$ and the root of $\tau_r(t)$  is labeled by $(z,y)$, where $z$ is some vertex in $V$. These vertices are called $s$-vertices.
			\item
			For each internal vertex $t \in V_T$  with label $(x,y)_p$, the root of $\tau_l(t)$ and $\tau_r(t)$  are labeled by $(x,y)$. These vertices are called $p$-vertices.
			\item
			For each internal vertex $t \in V_T$ that are labeled by $(x,z)_g$, the root of $\tau_l(t)$  is labeled by $(x,z)$  and the root of $\tau_r(t)$ is labeled by $(z,y)$ where $z$ is a vertex in $V$. These vertices are called $g$-vertices.
		\end{enumerate}
	\end{enumerate}
	
\end{definition}
Figure~\ref{parsetree-g} illustrates a binary parse tree for a GSP. Note that a binary parse tree for a GSP is not necessarily unique and can be computed by a linear time algorithm.
\begin{figure}[!h]
	\centering
	\includegraphics[width=14cm]{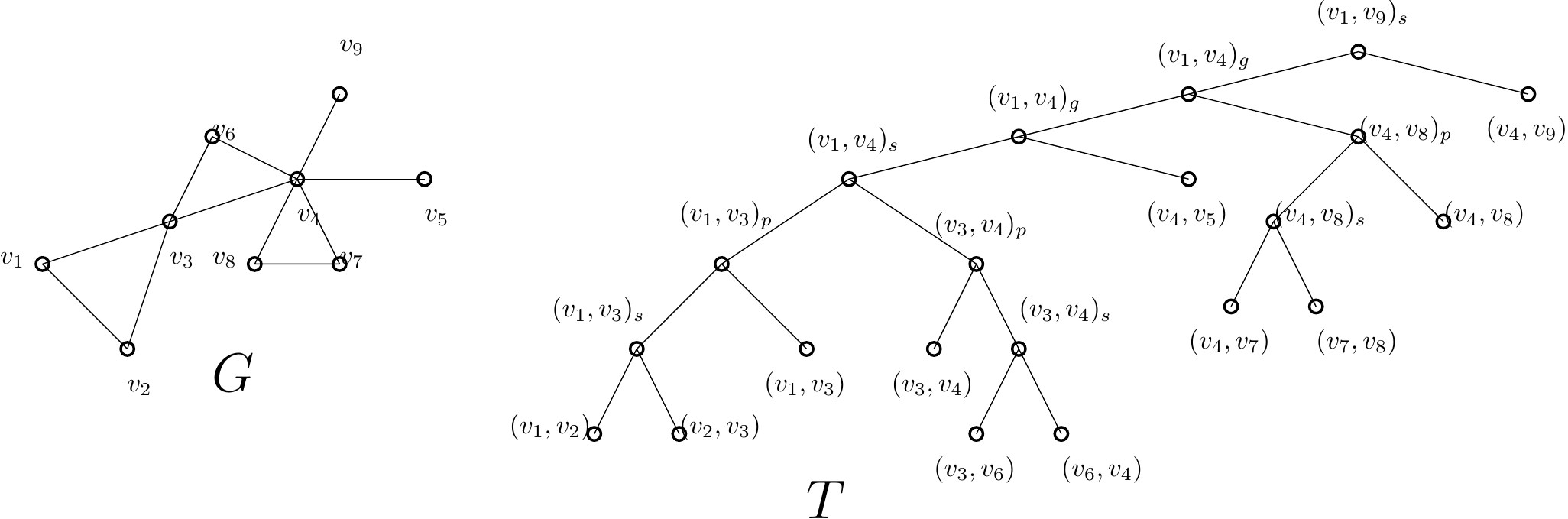}
	\caption{ A GSP graph and a binary parse tree for it.}\label{parsetree-g}
\end{figure}

\begin{lemma}[~\cite{hopcroft1973dividing}]
	For a given  GSP  $G$, a binary parse tree can be found in linear time.
\end{lemma}

\section{A dynamic programing algorithm for finding a minimum mixed dominating set}\label{Sec3}
We now define necessary notations that are used throughout this section. Then, we present our proposed algorithm and its procedures to find a $\gamma_{m}$-set, count them and $\gamma_m(G)$ for a given GSP $G$.

Let $t$ be a vertex in a parse tree $T$ corresponding to a GSP $G$ and $\hat{G}$ be a $p$-graph of subtree with root $t$. We define the sets $ch(t)$ and $\mathcal{MMD}_{i,j}(x,y)$ as follow:

$-$ The set $ch(t)$ consists of all children of $t$. In other words, in a  parse tree $T$, if $t$ is a leaf vertex, then  $ch(t)$ is an empty set and if it is an internal node, then $ch(t)$ contains two elements.

$-$ Let $(x,y)$ be the label of $t$ and $i,j \in \{0,1,2,3,4,5,6\}$. The set $\mathcal{MMD}_{i,j}(x,y)$ is an arbitrary $\gamma_m$-set  for $\hat{G}$  such that $(x,y)$ is the label of vertex $t$ and $i,j$ satisfy  one of the following conditions:

\begin{itemize}
	
	\item[\textbf{Case 0.}]
	 If $i=0$, then $x\in \mathcal{MMD}_{i,j}(x,y)$ and  at least one of its incident edges like $e$ are in $\mathcal{MMD}_{i,j}(x,y)$.
	 \item[\textbf{Case 1.}]
	 If $i=1$, then $x\in \mathcal{MMD}_{i,j}(x,y)$ and none of its incident edges are in $\mathcal{MMD}_{i,j}(x,y)$.
	\item[\textbf{Case 2.}]
	If $i=2$, then $x\notin \mathcal{MMD}_{i,j}(x,y)$ and at least one of its incident edges like $e$ are in $\mathcal{MMD}_{i,j}(x,y)$.
	\item[\textbf{Case 3.}]
	If $i=3$, then $x\notin \mathcal{MMD}_{i,j}(x,y)$ and  none of its incident edges are in $\mathcal{MMD}_{i,j}(x,y)$ since all of them are dominated by an edge or a vertex in $\mathcal{MMD}_{i,j}(x,y)$. Moreover, there is a vertex like $x'$ in $\mathcal{MMD}_{i,j}(x,y)$ such that $\{x,x'\} \in E(\hat{G})$.
	\item[\textbf{Case 4.}]
	If $i=4$, then $x\notin \mathcal{MMD}_{i,j}(x,y)$ and at least one of its incident edges are not dominated. Moreover, there is a vertex like $x'$ in $\mathcal{MMD}_{i,j}(x,y)$ such that $\{x,x'\} \in E(\hat{G})$.
	\item[\textbf{Case 5.}]
	If $i=5$, then $x\notin \mathcal{MMD}_{i,j}(x,y)$ and  none of its incident edges are in $\mathcal{MMD}_{i,j}(x,y)$ since all of them are dominated by an edge or a vertex in $\mathcal{MMD}_{i,j}(x,y)$. Moreover, there is no vertex like $x'$ in $\mathcal{MMD}_{i,j}(x,y)$ such that $\{x,x'\} \in E(\hat{G})$.
	\item[\textbf{Case 6.}]
	If $i=6$, then $x\notin \mathcal{MMD}_{i,j}(x,y)$ and none of its incident edges are in $\mathcal{MMD}_{i,j}(x,y)$ and at least one of them is not dominated. Moreover, there is not vertex like $x'$ in $\mathcal{MMD}_{i,j}(x,y)$ such that $\{x,x'\} \in E(\hat{G})$.
\end{itemize}
In a similar way, we can define these situations for $y$ based on $j$.

\newcommand{\Minsize}{\mbox{minSize}}
The function \Minsize\ is defined such that it receives a number of sets as input and returns a minimum sized set among them.

Using these definitions, our proposed algorithm constructs  $\mathcal{MMD}_{i,j}(x,y)$ as a mixed dominating set
with minimum cardinality for $G(x,y)$ and computes $\mathcal{N}_{i,j}(x,y)$ as the number of such sets.
At the end of algorithm, $\gamma_m$-set, $\gamma_m(G)$ and $\mathcal{N}_{\gamma_m}(G)$ are determined where $\gamma_m$-set is an smallest mixed dominating set and $\mathcal{N}_{\gamma_m}(G)$ is the number of mixed dominating sets with minimum size.

Now, we are ready to state our algorithm. The input of algorithm is a GSP  $G$. At
the first step of this algorithm, we find  a binary  parse tree like $T$ for $G$ by
some known linear time algorithms such as the one in~\cite{hopcroft1973dividing}. Next, our  algorithm traverses $T$ in bottom-up order.
Each subtree of parse tree  corresponds to a $p$-graph for $G$ and each vertex $t$ of $T$ is labeled by either $(x,y)_{i}$, $(x,y)_s$, $(x,y)_p$ or $(x,y)_g$ such that $x$ and $y$ are terminals in the corresponding $p$-graph of $\tau(t)$.
 
 For each visiting vertex $t$, one of the procedures ProcessLeaf, ProcessSvertex,
ProcessPvertex and ProcessGvertex is called. Input to each procedure consists of $x$ and $y$. 
 By traversing parse tree $T$ and visiting nodes of $T$ and calling proper procedures,
we find a subset of $\mathcal{MMD}_{i,j}(x,y)\subseteq V(\hat{G})$ such that for each
$i,j \in \{0,1,2,3,4,5,6\}$, $\mathcal{MMD}_{i,j}(x,y)$ stores a minimum mixed
dominating set of $\hat{G}$ with the assumption that $x$, $y$ or some of their incident edges cannot be dominated. 
 
 After visiting the root $r$ of $T$ and computing $\mathcal{MMD}_{i,j}(x,y)$ for $r$, a $\gamma_m$-set for $G$ can be found.
  It is enough to return a $\mathcal{MMD}_{i,j}(x,y)$ with minimum cardinality when {$i,j\in \{0,1,2,3\}$}.

The input of the ProcessLeaf procedure is a leaf vertex $v \in V_T$ which is labeled by $ (x,y)_{leaf}$ and its output is a set  $\mathcal{MMD}_{i,j}(x,y)$ for $i,j \in \mathcal{M}$. Note that a leaf corresponds  to an edge $\{x,y\}$ in $G$. 
For different $i,j\in \mathcal{M}$, we have summarized all valid cases as follows:

\begin{itemize}
	\item[1.]
	The vertices $x$, $y$ and the edge $\{x,y\}$ are dominated and at least one of them is a member of $\mathcal{MMD}_{i,j}(x,y)$.  So  $i,j$ satisfies one of following condition:
	\let\labelitemi\labelitemii
	\begin{itemize}
		\item 		$i=0$ and $j\in \{0,2\}$,
		\item
		 $i=1$ and $j\in\{1,3\}$,
		 \item
		  $i=2$ and $j\in \{0,2\}$,
		   \item
		   $i=3$ and $ j=1$.
	\end{itemize} 
	\item[2.]
The vertices $x$, $y$ and edge  $\{x,y\}$ are not dominated and  are not  members of $\mathcal{MMD}_{i,j}(x,y)$,  so we have  $i=j=6$.

	\end{itemize} 

Let $v$ is a vertex of $T$  labeled by $(x,y)_s$. In the ProcessSvertex procedure, we compute the set $\mathcal{MMD}_{i,j}(x,y)$ 
for given terminal vertices $x$, $y$ and common vertex $z$.

The sets $\mathcal{MMD}_{i_\ell,j_\ell}^\ell(x,z)$ and $\mathcal{MMD}_{i_r,j_r}^r(z,y)$ are the sets in correspondence to  $\tau_l(t)$ and $\tau_r(t)$ where the roots of $\tau_l(t)$ and $\tau_r(t)$ are labeled by $(x,z)$ and $(z,y)$, respectively.
We know that members of  $\mathcal{MMD}_{i_\ell,j_\ell}^\ell(x,z)$ and $\mathcal{MMD}_{i_r,j_r}^r(z,y)$  are those vertices of $T$ which are corresponding to $p$-graphs $G_1=(V_1,E_1,x,z)$, $G_2=(V_2,E_2,z,y)$ and $\hat{G}=G_1o_s G_2=(\hat{V},\hat{E},x,y)$, respectively.

Whether or not $z$ belongs to $\mathcal{MMD}_{i,j}(x,y)$ and which vertex or edge dominates $z$, the cases that can occur are summarized in Table 3. To be precise, consider the following cases:

\begin{itemize}
	\item[\textbf{Case 0.}]
	Vertex $z$ and at least one of its incident edges belong to $\mathcal{MMD}_{i,j}(x,y)$. So, we have   $(j_{\ell},i_r)\in \{(0,0),(0,1),(1,0)\}$.	
	\item[\textbf{Case 1.}]
	Vertex $z \in \mathcal{MMD}_{i,j}(x,y)$ and non of its incident edges belong to   $\mathcal{MMD}_{i,j}(x,y)$  which implies that $j_{\ell}=i_r=1$.	
	\item[\textbf{Case 2.}]
	Vertex $z \notin \mathcal{MMD}_{i,j}(x,y)$ and  an edge incident to $z$ belong to $\mathcal{MMD}_{i,j}(x,y)$. So, we have  $(j_\ell, i_r)\in \{2\}\times\{2,3,4,5,6\}$,  or $(i_r,j_\ell)\in \{2\}\times\{2,3,4,5,6\})$.
	\item[\textbf{Case 3.}]
	Vertex $z$ and  its incident edges does  not belong to $\mathcal{MMD}_{i,j}(x,y)$. So, we have 
	
	  $(j_\ell,i_r)\in\{(3,3),(3,5),(5,3)\}$. 	
\end{itemize}

\begin{algorithm}[h!]
	\begin{algorithmic}[1]
		\Procedure{ProcessSvertex}{$x,z,y$}
		\ForAll{$i,j\in M$}
		\State $setlist\gets \emptyset;$
		
			\State $Min \gets \mathcal{MMD}_{i,0}^\ell(x,z)\cup \mathcal{MMD}_{0,j}^r(z,y)$
			
			\State $\mathcal{N}_{i,j}(x,y) \gets \mathcal{N}_{i,0}^\ell(x,z)\times \mathcal{N}_{0,j}^r(z,y)$
				
		\ForAll{$(j_\ell,i_r)\in \{0,1\}$}
		\State Add $\mathcal{MMD}_{i,j_\ell}^\ell(x,z)\cup \mathcal{MMD}_{i_r,j}^r(z,y)$ to $setlist$;
		\State ProcessCalNum ($Min$, $\mathcal{N}_{i,j}(x,y)$, $\mathcal{MMD}_{i,j_\ell}^\ell(x,z)$, $ \mathcal{MMD}_{i_r,j}^r(z,y)$,  $\mathcal{N}_{i,j_\ell}^\ell(x,z)$, $ \mathcal{N}_{i_r,j}^r(z,y)$)
		\EndFor
		\ForAll{$(i_r)\in \{2,3,4,5,6\}$}
		\State Add $\mathcal{MMD}_{i,2}^\ell(x,z)\cup \mathcal{MMD}_{i_r,j}^r(z,y)$ to $setlist$;
		\State ProcessCalNum ($Min$, $\mathcal{N}_{i,j}(x,y)$, $\mathcal{MMD}_{i,j_\ell}^\ell(x,z)$, $ \mathcal{MMD}_{i_r,j}^r(z,y)$,  $\mathcal{N}_{i,j_\ell}^\ell(x,z)$, $ \mathcal{N}_{i_r,j}^r(z,y)$)
		\EndFor
		\ForAll{$(j_\ell)\in \{2,3,4,5,6\}$}
		\State Add $\mathcal{MMD}_{i,j_\ell}^\ell(x,z)\cup \mathcal{MMD}_{2,j}^r(z,y)$ to $setlist$;
	\State ProcessCalNum ($Min$, $\mathcal{N}_{i,j}(x,y)$, $\mathcal{MMD}_{i,j_\ell}^\ell(x,z)$, $ \mathcal{MMD}_{i_r,j}^r(z,y)$, 
		  $\mathcal{N}_{i,j_\ell}^\ell(x,z)$, $ \mathcal{N}_{i_r,j}^r(z,y)$)
		\EndFor
		\State Add $\mathcal{MMD}_{i,3}^\ell(x,z)\cup \mathcal{MMD}_{3,j}^r(z,y)$ to $setlist$;
		\State ProcessCalNum ($Min$, $\mathcal{N}_{i,j}(x,y)$, $\mathcal{MMD}_{i,j_\ell}^\ell(x,z)$, $ \mathcal{MMD}_{i_r,y_j}^r(z,y)$,  $\mathcal{N}_{i,j_\ell}^\ell(x,z)$, $ \mathcal{N}_{i_r,j}^r(z,y)$)
		\State Add $\mathcal{MMD}_{i,3}^\ell(x,z)\cup \mathcal{MMD}_{5,j}^r(z,y)$ to $setlist$;
		\State ProcessCalNum ($Min$, $\mathcal{N}_{i,j}(x,y)$, $\mathcal{MMD}_{i,j_\ell}^\ell(x,z)$, $ \mathcal{MMD}_{i_r,j}^r(z,y)$,  $\mathcal{N}_{i,j_\ell}^\ell(x,z)$, $ \mathcal{N}_{i_r,j}^r(z,y)$)
		\State Add $\mathcal{MMD}_{i,5}^\ell(x,z)\cup \mathcal{MMD}_{3,j}^r(z,y)$ to $setlist$;
	\State ProcessCalNum ($Min$, $\mathcal{N}_{i,j}(x,y)$, $\mathcal{MMD}_{i,j_\ell}^\ell(x,z)$, $ \mathcal{MMD}_{i_r,j}^r(z,y)$,  $\mathcal{N}_{i,j_\ell}^\ell(x,z)$, $ \mathcal{N}_{i_r,j}^r(z,y)$)
		\State $\mathcal{MMD}_{i,j}(x,y)\gets \Minsize(setlist);$
		\EndFor
		\EndProcedure
	\end{algorithmic}
\end{algorithm}

Now, let $v$ be a vertex of $T$ labeled by $(x,y)_p$.  The sets corresponding to  $\tau_l(t)$ and $\tau_r(t)$ are $\mathcal{MMD}_{i_\ell,j_\ell}^\ell(x,y)$ and $\mathcal{MMD}_{i_r,j_r}^r(x,y)$, respectively.

For each $i,j \in M$, we describe a method to find $\mathcal{MMD}_{i,j}(x,y)$. Note that it is enough to find a relation between values of $(i,j)$,  $(i_\ell,j_\ell)$ and $(i_r,j_r)$.
To do so, we use the procedure FindList.
  Let the input of this procedure be a value like $i \in M$. Then,  the procedure returns a set of pairs which  are proper values for $i_\ell$ and $i_r$.  Note that for $j \in M$, the procedure returns proper  $j_\ell$ and $j_r$, similarly.

\begin{algorithm}[h!]
	\begin{algorithmic}[1] 
		\Procedure{FindList}{$k$}
		\State $M \gets \{0,1,2,3,4,5,6\}$
		\ForAll{$k\in M$}
		\State $list\gets \emptyset$
		\EndFor
		\Switch{$k$}
		\Case{$0$}
		\State $list\gets \{(0,0),(0,1),(1,0)\}$
		\EndCase
		\Case{$1$}
		\State $list\gets \{(1,1)\}$
		\EndCase
		\Case{$2$}
		\ForAll{$k'\in M \setminus\{0,1\}$}
		\State Add $(2,k')$ to $list$
		\State Add $(k',2)$ to $list$
		\EndFor
		\EndCase
		\Case{$3$}
		\State $list\gets \{(3,3),(3,5),(5,3)\}$
		\EndCase
		\Case{$4$}
		\ForAll{$k'\in M \setminus\{0,1,2\}$}
		\State Add $(4,k')$ to $list$
		\State Add $(k',4)$ to $list$
		\EndFor
		\EndCase
		\Case{$5$}
		\State Add $(5,5)$ to $list$
		\EndCase
		\Case{$6$}
		\State Add $(5,6),(6,5),(6,6)$ to $list$
		\EndCase	
		\EndSwitch
		\EndProcedure
	\end{algorithmic}
\end{algorithm}

\begin{algorithm}[h!]
	\begin{algorithmic}[1]
		\Procedure{ProcessPvertex}{$x,y$}
		\ForAll{$i,j\in M$}
		\State $list1\gets  ProcessFindlist(i)$
		\State $list2 \gets ProcessFindlist(j)$
		\State $(i_{\ell_1},j_{\ell_1}),(i_{r_1},j_{r_1}) \gets \text{an arbitrary element of } list1 \times list2$
		\State $Min \gets \mathcal{MMD}_{i_{\ell_1},j_{\ell_1}}^\ell(x,y)\cup \mathcal{MMD}_{i_{r_1},j_{r_1}}^r(x,y)$
		\State $\mathcal{N}_{i,j}(x,y) \gets \mathcal{N}_{i_{\ell_1},j_{\ell_1}}^\ell(x,y)\times \mathcal{N}_{i_{r_1},j_{r_1}}^r(x,y)$
		\ForAll{$(i_{\ell},j_{\ell}),(i_r,j_{r}) \in list1 \times list2$}
		\State Add $\mathcal{MMD}_{i_{\ell_1},j_{\ell_1}}^\ell(x,y)\cup \mathcal{MMD}_{i_r,j_r}^r(x,y)$ to $setlist$;
		\State ProcessCalNum ($Min$, $\mathcal{N}_{i,j}(x,y)$, $\mathcal{MMD}_{i_\ell,j_\ell}^\ell(x,y)$, $ \mathcal{MMD}_{i_r,j_r}^r(x,y)$,  $\mathcal{N}_{i_\ell,j_\ell}^\ell(x,y)$, $ \mathcal{N}_{i_r,j_r}^r(x,y)$)
		\EndFor
		\State $\mathcal{MMD}_{i,j}(x,y)\gets \Minsize(setlist);$
		\EndFor
		\EndProcedure
	\end{algorithmic}
\end{algorithm}

 Note that $\tau_l(t)$, $\tau_r(t)$ and $\tau(t)$   correspond  to $p$-graphs $G_1=(V_1,E_1,x,y)$, $G_2=(V_2,E_2,x,y)$ and $\hat{G}=G_1 o_pG_2=(\hat{V},\hat{E},x,y)$, respectively.
   Now, for $i\in \mathcal{M}$ (resp. $j\in \mathcal{M}$) the values of $i_\ell$ and $i_r$ (resp. $j_\ell$ and $j_r$)  are determined  as follows (they are also shown in Table~\ref{pvert}).

 \begin{itemize}
 \item[\textbf{Case 0.}]
$i=0$ implies  $(i_{\ell},i_r)\in \{(0,0),(0,1),(1,0)\}$
  \item[\textbf{Case 1.}]
 $i=1$ implies $i_{\ell}=i_r=1$,
\item[\textbf{Case 2.}]
$i=2$ implies  $(i_{\ell},i_r)\in \{2\}\times\{2,3,4,5,6\}$  or $(i_{\ell},i_r)\in \{2,3,4,5,6\}\times \{2\}$,
\item[\textbf{Case 3.}]
$i=3$ implies  $(i_{\ell},i_r)\in \{(3,3),(3,5),(5,3)\}$

\item[\textbf{Case 4.}]
$i=4$ implies  $(i_{\ell},i_r)\in \{(3,4), (3,6),(4,4),(4,5),(4,6)\}$ or $(i_r,i_\ell)\in \{(3,4), (3,6),(4,4),(4,5),(4,6)\}$ 

\item[\textbf{Case 5.}]
$i=5$ implies  $i_{\ell}=i_r=5$.

\item[\textbf{Case 6.}]
$i=6$ implies    $(i_{\ell},i_r)\in \{(5,6),(6,5),(6,6)\}$.
 \end{itemize}

Let $v$ be a vertex of $T$ which is labeled by $(x,y)_g$. In  procedure ProcessGvertex,
the set $\mathcal{MMD}_{i,j}(x,y)$ is computed for the given vertices $x$ and $y$. 
The sets corresponding to  $\tau_l(t)$ and $\tau_r(t)$ are $\mathcal{MMD}_{i,j_\ell}^\ell(x,y)$ and $\mathcal{MMD}_{i_r,j_r}^r(x,y)$, respectively.

\begin{algorithm}[h!]
	\begin{algorithmic}[1]
		\Procedure{ProcessGvertex}{$x,y$}
		\ForAll{$i,j\in M$}
		\State $list1\gets  ProcessFindlist(j);$
					\State $j_{\ell_1},i_{r_1} \gets \text{an arbitrary element of } list1$
			\State $Min \gets \mathcal{MMD}_{i,j_{\ell_1}}^\ell(x,z)\cup \mathcal{MMD}_{i_{r_1},j_r}^r(z,y)$
						\State $\mathcal{N}_{i,j}(x,y) \gets \mathcal{N}_{i,j_{\ell_1}}^\ell(x,z)\times \mathcal{N}_{i_{r_1},j_r}^r(z,y)$
				\ForAll{$((j_{\ell},i_r),j_{r}) \in list1 \times \{0,1,2,3\}$}
		\State Add $\mathcal{MMD}_{i,j_\ell}^\ell(x,y)\cup \mathcal{MMD}_{i_r,j_r}^r(x,y)$ to $setlist$;
		\State ProcessCalNum ($Min$, $\mathcal{N}_{i,j}(x,y)$, $\mathcal{MMD}_{i,j_\ell}^\ell(x,z)$, $ \mathcal{MMD}_{i_r,j_r}^r(x,y)$,  $\mathcal{N}_{i,j_\ell}^\ell(x,z)$, $ \mathcal{N}_{i_r,j_r}^r(x,y)$)
		\EndFor
				\State $\mathcal{MMD}_{i,j}(x,y)\gets \Minsize(setlist);$
		\EndFor
		\EndProcedure
	\end{algorithmic}
\end{algorithm}
             
Let the roots of $\tau_l(t)$ and $\tau_r(t)$ be labeled by $(x,y)$ and $(y,z)$, respectively, for some $z \in V$ and  $\mathcal{MMD}_{i_\ell,j_\ell}^\ell(x,y)$ and $\mathcal{MMD}_{i_r,j_r}^r(x,y)$ are the  associated ones with the vertices $(x,y)$ and $(y,z)$ of $T$.
Obviously, $z$ does not appear in ancestors of $t$ in parse tree. So, $z$ and all of its incident edges must be closely dominated which implies $j\in\{0,1,2,3\}$. Since $y$ is the common vertex between $G_1$ and $G_2$, according to $j$, the set denoted by \textit{list1} which equals the set of  possible pairs can be computed for $j_{\ell}$ and $i_r$ by procedure Findlist. Several  cases are possible for $y$ which are shown in Table~\ref{gvert} and are discussed below:

\begin{itemize}
	\item[\textbf{Case 0.}]
		$j=0$ implies  $(j_{\ell},i_r)\in\{(0,0),(1,0),(0,1)\}$.
	\item[\textbf{Case 1.}]
	$j=1$ implies  $j_{\ell}=i_r=1$.
	
	\item[\textbf{Case 2.}]
	$j=2$ implies   $(j_{\ell},i_r)\in \{2\}\times \{2,3,4,5,6\}$  or $(j_{\ell},i_r)\in \{2,3,4,5,6\}\times\{2\}$.
	\item[\textbf{Case 3.}]
	$j=3$ implies   $(j_{\ell},i_r)\in\{(3,3),(3,5),(5,3)\}$.
	\item[\textbf{Case 4.}]
	$j=4$ implies    $(j_{\ell},i_r)\in \{3\}\times \{4,6\}$, $(j_{\ell},i_r)\in \{4\}\times \{4,5,6\}$, $(i_r,j_{\ell})\in \{3\}\times \{4,6\}$ or $(i_r,j_{\ell})\in \{4\}\times \{4,5,6\}$.  
		\item[\textbf{Case 5.}]
	$j=5$ implies   $j_{\ell}=i_r=5$.
	\item[\textbf{Case 6.}]
	$j=6$ implies   $(j_{\ell},i_r)\in\{(5,6),(6,5),(6,6)\}$.
\end{itemize}

By $\mathcal{MMD}_{i,j}(x,y)\leftarrow \Minsize(Setlist)$, we remove all of undefinable sets from $Setlist$. If $Setlist$ is empty, then 
$\mathcal{MMD}_{i,j}(x,y)$ becomes undefinable.

\begin{algorithm}[h!]
\begin{algorithmic}[1]
	\Procedure{ProcessCalNum}{$Min$,$N$, $S_1$, $S_2$, $s'_1$, $s'_2$}
		\If{$\vert S_1\cup S_2 \vert \leq Min$}
	\State $N \gets s'_1 \times s'_2$
	\ElsIf {$\vert S_1 \cup S_2 \vert= Min$}
	\State $N \gets N+ s'_1 \times s'_2$
	\EndIf
	\EndProcedure
\end{algorithmic}
\end{algorithm}

\begin{algorithm}[h!]
	\begin{algorithmic}[1]
\State Find a parse tree of $G$ like $T$
		\For  { each $v$ in a post order traverse of parse tree}
		\Switch{$type\;of\;v$}
		\Case{$Leaf$}
		\State ProcessLeaf$(x,y)$
		\Comment{$(x,y)_{i}$ is the label of $v$}
		\EndCase
		\Case{$s-vertex$}
		\State ProcessSvertex$(x,z,y)$
		\Comment{$(x,y)_s$, $(x,z)$ and $(z,y)$ are labels of $v$, left and right child of $v$, respectively. }
		\EndCase
				\Case{$p-vertex$}
		\State ProcessPvertex$(x,y)$
		\Comment{$(x,y)_p$ is the label of $v$ and the labels of left and right child of $v$  is sequels $(x,y)$. }
		\EndCase
				\Case{$g-vertex$}
		\State ProcessGvertex$(x,y,z)$
		\Comment{$(x,y)_s$, $(x,y)$ and $(y,z)$ are labels of $v$, left and right child of $v$, respectively.}
		\EndCase
		\EndSwitch
		\EndFor
		\State $D\gets \emptyset$
		\State $Min\gets \mathcal{MMD}_{0,0}(x,y)$
		\ForAll{$i,j\in \{0,1,2,3\}$}
		\State Add $\mathcal{MMD}_{i,j}(x,y)$ to $D$
		\If{$\vert\mathcal{MMD}_{i,j}(x,y) \vert \leq Min$}
				\State $N_{\gamma_{m}} \gets N(x_{i},y_{j})$
		\ElsIf {$\vert \mathcal{MMD}_{i,j}(x,y) \vert= Min$}
		\State $N_{\gamma_{m}}  \gets N_{\gamma_{m}} + N(x_{i},y_{j})$
		\EndIf
				\EndFor
		\State $\gamma_{m}$-set$\gets \Minsize(D)$
		\State $\gamma_{m}(G) \gets \vert \gamma_{m}$-set $\vert$
	\end{algorithmic}
\end{algorithm}

\begin{algorithm}[h!]
\caption{:Finding a $\gamma_{m}$-sets of a GSP}\label{findset}
	\begin{algorithmic}[1]
		\Procedure{ProcessLeaf}{$x,y$}
		\ForAll{$i,j\in \{0,1,2,3,4,5,6\}$ }
		\State $\mathcal{MMD}_{i,j}(x,y) \gets NaN; $
		\State $\mathcal{N}_{i,j}(x,y)=0; $
		\EndFor
		\State $\mathcal{MMD}_{0,0}(x,y)\gets \{x,y,xy\} $ and  $\mathcal{N}_{0,0}(x,y)=1; $
		\State $\mathcal{MMD}_{0,2}(x,y)\gets \{x,xy\} $ and $\mathcal{N}_{0,2}(x,y)=1; $
		\State $\mathcal{MMD}_{1,1}(x,y)\gets \{x,y\} $ and  $\mathcal{N}_{1,1}(x,y)=1; $
		\State $\mathcal{MMD}_{1,3}(x,y)\gets \{x\}$ and $\mathcal{N}_{1,3}(x,y)=1; $
		\State $\mathcal{MMD}_{2,0}(x,y)\gets \{y,xy\} $ and $\mathcal{N}_{2,0}(x,y)=1; $
		\State $\mathcal{MMD}_{2,2}(x,y)\gets \{xy\} $ and $\mathcal{N}_{2,2}(x,y)=1; $
		\State $\mathcal{MMD}_{3,1}(x,y)\gets \{y\} $ and $\mathcal{N}_{3,1}(x,y)=1; $
		\State $\mathcal{MMD}_{6,6}(x,y)\gets \emptyset $ and $\mathcal{N}_{6,6}(x,y)=1. $
				\EndProcedure
	\end{algorithmic}
\end{algorithm}

The correctness and complexity of our proposed algorithm is discussed below.

\begin{theorem}\label{Complexity}
	For a given generalized series-parallel graph $G=(V,E)$, Algorithm~\ref{findset} finds a $\gamma_m$-set for $G$ in time $O(\vert V \vert)$.
	\end{theorem}	

\begin{proof}
In Algorithm ~\ref{findset}, we traverse the parse tree $T$ in a bottom-up fashion and computes at most $49$ sets for each internal vertex of them.
 Each initial set for leaves of the tree represents all possible mixed dominating sets in a graph consisting of only one edge. 
 Let $G_1 = (V_1, E_1)$ and $G_2=(V_2, E_2)$ be the graphs represented by the subtrees
$\tau_l(t)$ and $\tau_r(t)$. Assume that they are input to procedures ProcessLeaf,
ProcessSvertex,   ProcessPvertex and ProcessGvertex. It is easy to see that these
procedures   find all possible $\gamma_m$-sets in each corresponding graph. Finally,
our algorithm extracts only a valid minimum mixed dominating set. These steps of
algorithm require at most $\vert O(V_T) \vert$ operations. Since each binary tree with
$n$ leaves has $O(n)$ vertices and the binary parse tree of every GSP graph  has
$\vert E(G) \vert$ leaves, so $\vert V_T\vert \in O(\vert E (G)\vert)$. Every GSP
graph is planar, and in a planar graph we have $\vert E \vert \leq 3 \vert V \vert -
6$. Also, we know that a parse tree $T$ can be constructed in $O(\vert V \vert)$~\cite{kikuno1983linear}. So, the algorithm  computes a $\gamma_{m}$-set for a given GSP graph $G$ in time $O(\vert V \vert)$.
\end{proof}

\begin{table}
	\centering
	\caption{Different situations for  $s$-vertices. }\label{T3}
	\scalebox{.8}{
		\begin{tabular}{ r|c|c|c|c|c|c|c| }
			\multicolumn{1}{r}{}
			&  \multicolumn{1}{c}{Case}
			&  \multicolumn{1}{c}{(a)}
			&  \multicolumn{1}{c}{(b)}
			& \multicolumn{1}{c}{(c)} 
			&  \multicolumn{1}{c}{(d)}
			&  \multicolumn{1}{c}{(e) }
			\\
			\cline{2-7}
			& 0 & \begin{minipage}{.20\textwidth}
				\vspace{3mm}
				\includegraphics[width=\linewidth, height=25mm]{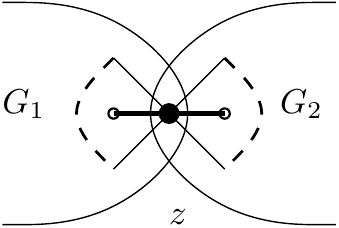}\vspace{3mm}
			\end{minipage} & \begin{minipage}{.20\textwidth}
			\includegraphics[width=\linewidth, height=25mm]{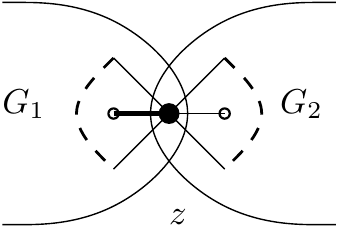}
		\end{minipage} & \begin{minipage}{.20\textwidth}
		\includegraphics[width=\linewidth, height=25mm]{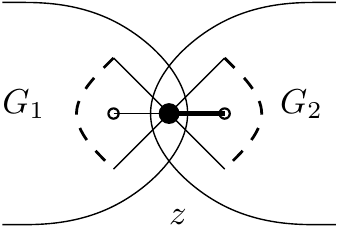}
	\end{minipage} &  &    \\
	\cline{3-7}
	&  & $j_{\ell}=0,i_r=0$ & $j_{\ell}=0,i_r=1$ & $j_{\ell}=1,i_r=0$ & &    \\
	\cline{2-7}
	& 1 & \begin{minipage}{.20\textwidth}
		\vspace{3mm}
		\includegraphics[width=\linewidth, height=25mm]{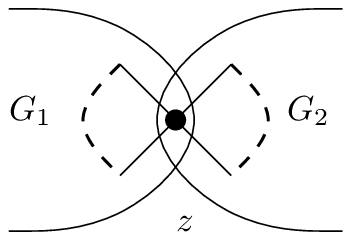}\vspace{3mm}
	\end{minipage} &  &  &  &  \\
	\cline{3-7}
	&  & $j_{\ell}=1,i_r=1$  &  &  &  &  \\
	\cline{2-7}
	& 2 & \begin{minipage}{.20\textwidth}
		\vspace{3mm}
		\includegraphics[width=\linewidth, height=25mm]{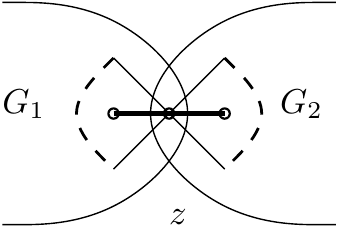}\vspace{3mm}
	\end{minipage} & \begin{minipage}{.20\textwidth}
	\includegraphics[width=\linewidth, height=25mm]{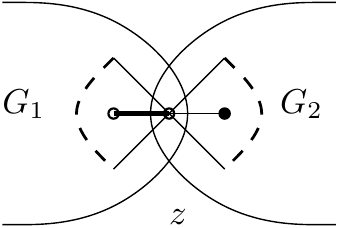}
\end{minipage} & \begin{minipage}{.20\textwidth}
\includegraphics[width=\linewidth, height=25mm]{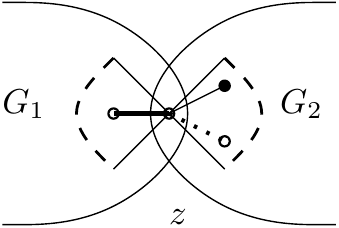}
\end{minipage} & \begin{minipage}{.20\textwidth}
\includegraphics[width=\linewidth, height=25mm]{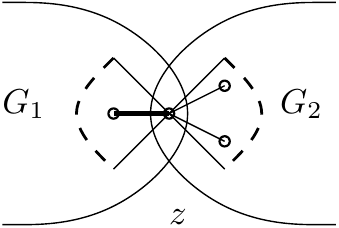}
\end{minipage} & \begin{minipage}{.20\textwidth}
\includegraphics[width=\linewidth, height=25mm]{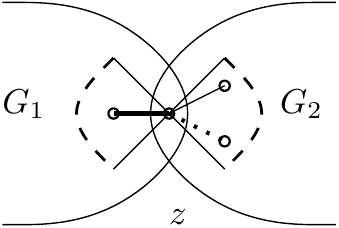}
\end{minipage}    \\
\cline{3-7}
&  & $j_{\ell}=2,i_r=2$ & $j_{\ell}=2,i_r=3$ & $j_{\ell}=2,i_r=4$ & $j_{\ell}=2,i_r=5$ & $j_{\ell}=2,i_r=6$   \\
\cline{3-7}
&  & & $j_{\ell}=3,i_r=2$ & $j_{\ell}=4,i_r=2$ & $j_{\ell}=5,i_r=2$ & $j_{\ell}=6,i_r=2$   \\
\cline{2-7}
& 3 & \begin{minipage}{.20\textwidth}
	\vspace{3mm}
	\includegraphics[width=\linewidth, height=25mm]{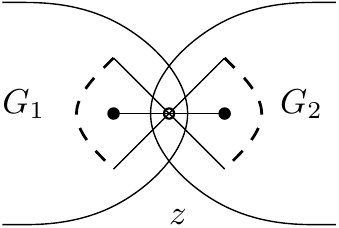}\vspace{3mm}
\end{minipage} & \begin{minipage}{.20\textwidth}
\includegraphics[width=\linewidth, height=25mm]{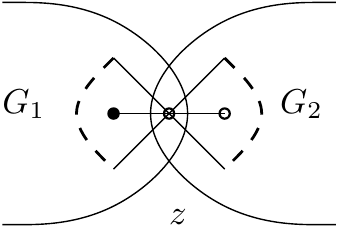}
\end{minipage} & \begin{minipage}{.20\textwidth}
\includegraphics[width=\linewidth, height=25mm]{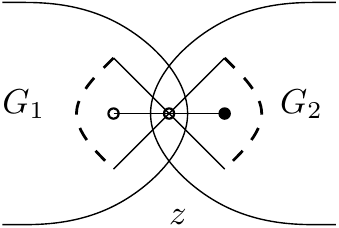}\label{test}
\end{minipage} &  &    \\
\cline{3-7}
&  & $j_{\ell}=3,i_r=3$ & $j_{\ell}=3,i_r=5$ & $j_{\ell}=5,i_r=3$ & &    \\
\cline{2-7}
\end{tabular}	}
\end{table}

\begin{table}
	\centering
	\caption{Different situations for  $p$-vertices.}\label{pvert}
	\scalebox{0.7}{  
		\begin{tabular}{ r|c|c|c|c|c|c|c|c| }
			\multicolumn{1}{r}{}
			&  \multicolumn{1}{c}{Case}
			&  \multicolumn{1}{c}{(a)}
			&  \multicolumn{1}{c}{(b)}
			& \multicolumn{1}{c}{(c)} 
			&  \multicolumn{1}{c}{(d)}
			&  \multicolumn{1}{c}{(e) }
			\\
			\cline{2-7}
			& 0 & \begin{minipage}{.20\textwidth}\vspace{2mm}
				\includegraphics[width=\linewidth, height=30mm]{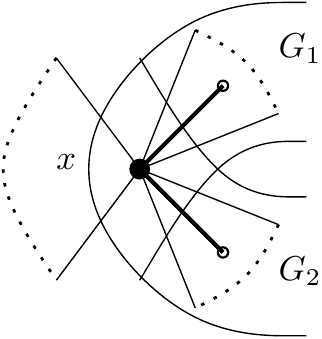}\vspace{2mm}
			\end{minipage} & \begin{minipage}{.20\textwidth}
			\includegraphics[width=\linewidth, height=30mm]{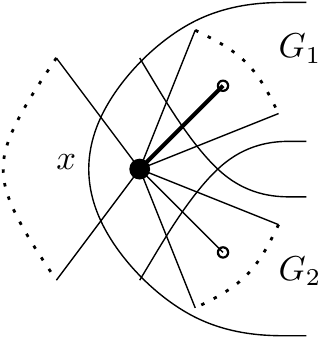}
		\end{minipage} & \begin{minipage}{.20\textwidth}
		\includegraphics[width=\linewidth, height=30mm]{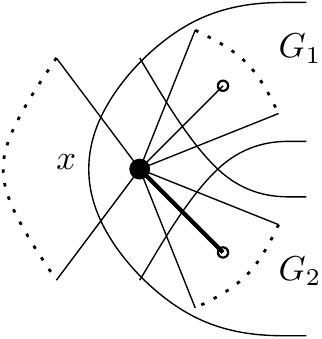}
	\end{minipage} &  &    \\
	\cline{3-7}
	&  & $j_{\ell}=0,i_r=0$ & $j_{\ell}=0,i_r=1$ & $j_{\ell}=1,i_r=0$ & &    \\
	\cline{2-7}
	& 1 & \begin{minipage}{.20\textwidth}\vspace{2mm}
		\includegraphics[width=\linewidth, height=30mm]{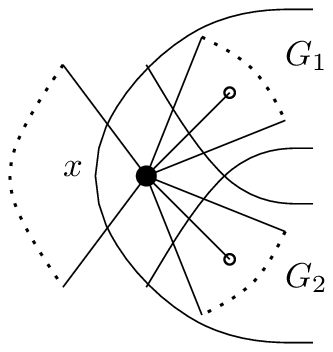}\vspace{2mm}
	\end{minipage} &  &  &  &  \\
	\cline{3-7}
	&  & $j_{\ell}=1,i_r=1$  &  &  &  &  \\
	\cline{2-7}
	& 2 & \begin{minipage}{.20\textwidth}\vspace{2mm}
		\includegraphics[width=\linewidth, height=30mm]{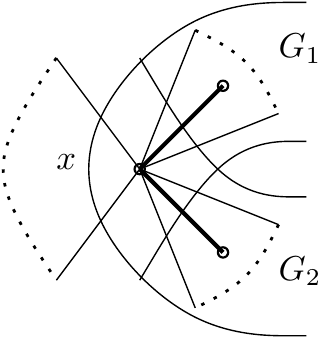}\vspace{2mm}
	\end{minipage} & \begin{minipage}{.20\textwidth}
	\includegraphics[width=\linewidth, height=30mm]{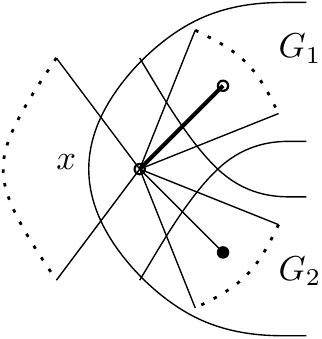}
\end{minipage} & \begin{minipage}{.20\textwidth}
\includegraphics[width=\linewidth, height=30mm]{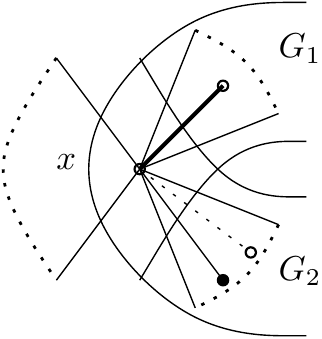}
\end{minipage} & \begin{minipage}{.20\textwidth}
\includegraphics[width=\linewidth, height=30mm]{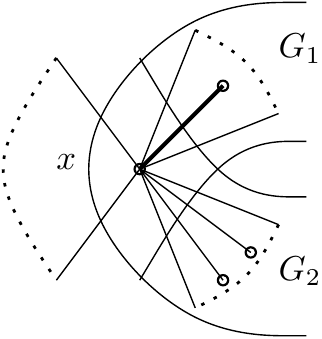}
\end{minipage} & \begin{minipage}{.20\textwidth}
\includegraphics[width=\linewidth, height=30mm]{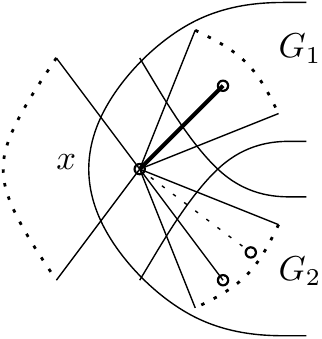}
\end{minipage}    \\
\cline{3-7}
&  & $j_{\ell}=2,i_r=2$ & $j_{\ell}=2,i_r=3$ & $j_{\ell}=2,i_r=4$ & $j_{\ell}=2,i_r=5$ & $j_{\ell}=2,i_r=6$   \\
\cline{3-7}
&  & & $j_{\ell}=3,i_r=2$ & $j_{\ell}=4,i_r=2$ & $j_{\ell}=5,i_r=2$ & $j_{\ell}=6,i_r=2$   \\
\cline{2-7}
& 3 & \begin{minipage}{.20\textwidth}\vspace{2mm}
	\includegraphics[width=\linewidth, height=30mm]{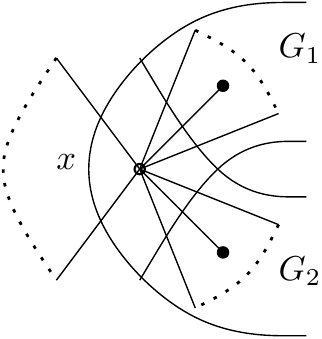}\vspace{2mm}
\end{minipage} & \begin{minipage}{.20\textwidth}
\includegraphics[width=\linewidth, height=30mm]{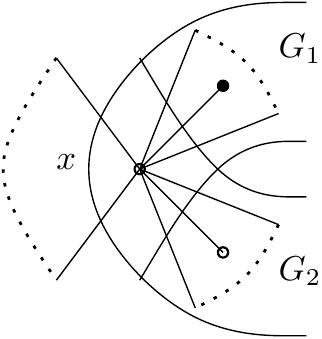}
\end{minipage} & \begin{minipage}{.20\textwidth}
\includegraphics[width=\linewidth, height=30mm]{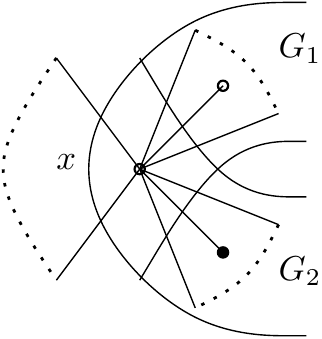}
\end{minipage} &  &    \\
\cline{3-7}
&  & $j_{\ell}=3,i_r=3$ & $j_{\ell}=3,i_r=5$ & $j_{\ell}=5,i_r=3$ & &    \\
\cline{2-7}
& 4 & \begin{minipage}{.20\textwidth}\vspace{2mm}
	\includegraphics[width=\linewidth, height=30mm]{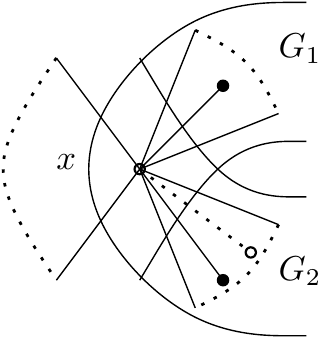}\vspace{2mm}
\end{minipage} & \begin{minipage}{.20\textwidth}
\includegraphics[width=\linewidth, height=30mm]{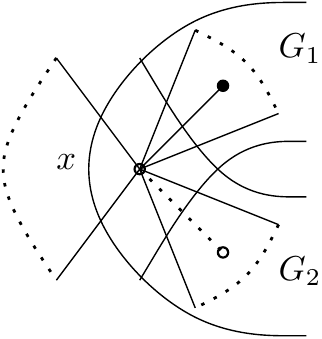}
\end{minipage} & \begin{minipage}{.20\textwidth}
\includegraphics[width=\linewidth, height=30mm]{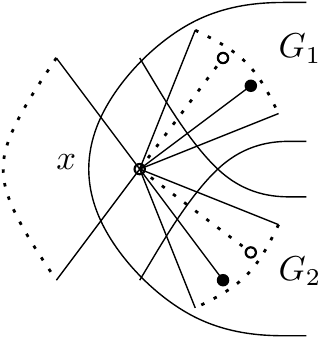}
\end{minipage} & \begin{minipage}{.20\textwidth}
\includegraphics[width=\linewidth, height=30mm]{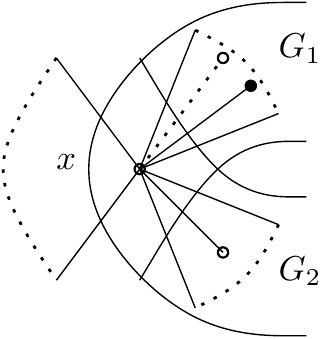}
\end{minipage} & \begin{minipage}{.20\textwidth}
\includegraphics[width=\linewidth, height=30mm]{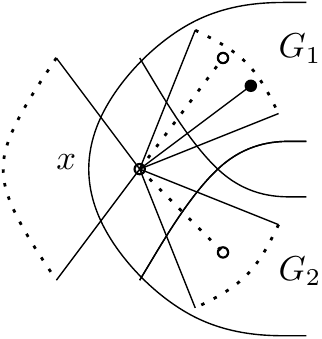}
\end{minipage}    \\
\cline{3-7}
&  & $j_{\ell}=3,i_r=4$ & $j_{\ell}=3,i_r=6$ & $j_{\ell}=4,i_r=4$ & $j_{\ell}=4,i_r=5$ & $j_{\ell}=4,i_r=6$   \\
\cline{3-7}
&  & $j_{\ell}=4,i_r=3$ & $j_{\ell}=6,i_r=3$ &  & $j_{\ell}=5,i_r=4$ & $j_{\ell}=6,i_r=4$   \\
\cline{2-7}
& 5 & \begin{minipage}{.20\textwidth}\vspace{2mm}
	\includegraphics[width=\linewidth, height=30mm]{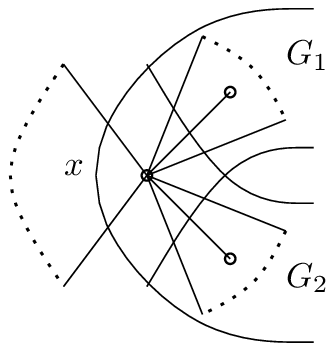}\vspace{2mm}
\end{minipage} &  &  &  &  \\
\cline{3-7}
&  & $j_{\ell}=5,i_r=5$  &  &  &  &  \\

\cline{2-7}
& 6 & \begin{minipage}{.20\textwidth}\vspace{2mm}
	\includegraphics[width=\linewidth, height=30mm]{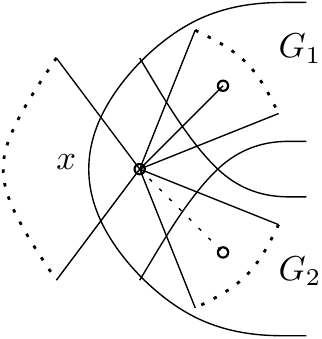}\vspace{2mm}
\end{minipage} & \begin{minipage}{.20\textwidth}
\includegraphics[width=\linewidth, height=30mm]{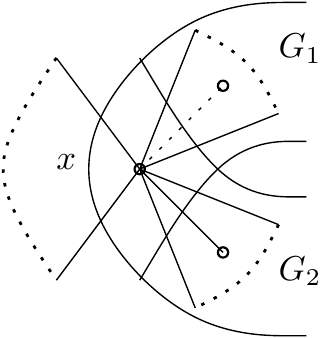}
\end{minipage} & \begin{minipage}{.20\textwidth}
\includegraphics[width=\linewidth, height=30mm]{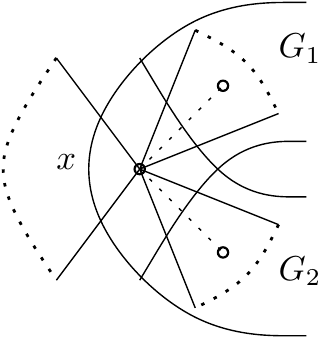}
\end{minipage} &  &    \\
\cline{3-7}
&  & $j_{\ell}=6,i_r=5$ & $j_{\ell}=5,i_r=6$ & $j_{\ell}=6,i_r=6$ &  &    \\
\cline{2-7}
\end{tabular}
}\label{pver0-6}
\end{table}

\begin{table}\label{gver0-3}
	\centering
	\caption{Different situations for  $g$-vertices.}\label{gvert}
	\scalebox{0.7}{  
		\begin{tabular}{ r|c|c|c|c|c|c|c|c| }
			\multicolumn{1}{r}{}
			&  \multicolumn{1}{c}{Case}
			&  \multicolumn{1}{c}{(a)}
			&  \multicolumn{1}{c}{(b)}
			& \multicolumn{1}{c}{(c)} 
			&  \multicolumn{1}{c}{(d)}
			&  \multicolumn{1}{c}{(e) }
			\\
			\cline{2-7}
			& 0 & \begin{minipage}{.20\textwidth}\vspace{2mm}
				\includegraphics[width=\linewidth, height=25mm]{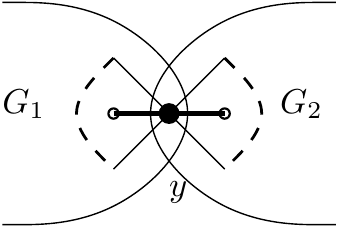}\vspace{2mm}
			\end{minipage} & \begin{minipage}{.20\textwidth}
			\includegraphics[width=\linewidth, height=25mm]{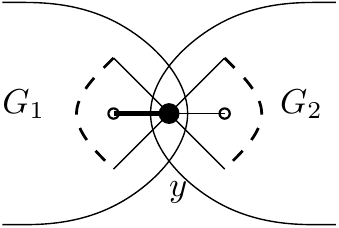}
		\end{minipage} & \begin{minipage}{.20\textwidth}
		\includegraphics[width=\linewidth, height=25mm]{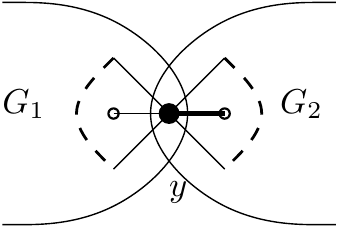}
	\end{minipage} &  &    \\
	\cline{3-7}
	&  & $j_{\ell}=0,i_r=0$ & $j_{\ell}=0,i_r=1$ & $j_{\ell}=1,i_r=0$ & &    \\
	\cline{2-7}
	& 1 & \begin{minipage}{.20\textwidth}\vspace{2mm}
		\includegraphics[width=\linewidth, height=25mm]{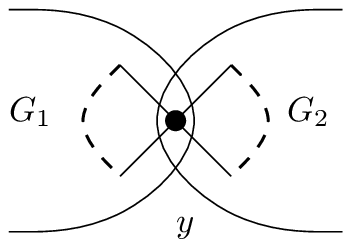}\vspace{2mm}
	\end{minipage} &  &  &  &  \\
	\cline{3-7}
	&  & $j_{\ell}=1,i_r=1$  &  &  &  &  \\
	\cline{2-7}
	& 2 & \begin{minipage}{.20\textwidth}\vspace{2mm}
		\includegraphics[width=\linewidth, height=25mm]{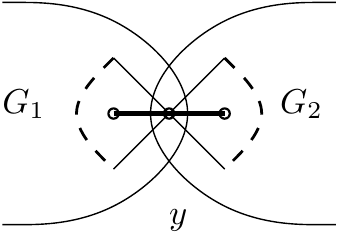}\vspace{2mm}
	\end{minipage} & \begin{minipage}{.20\textwidth}
	\includegraphics[width=\linewidth, height=25mm]{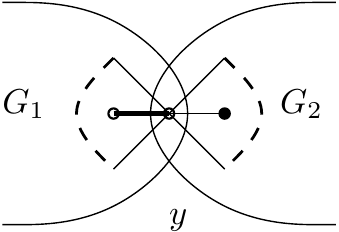}
\end{minipage} & \begin{minipage}{.20\textwidth}
\includegraphics[width=\linewidth, height=25mm]{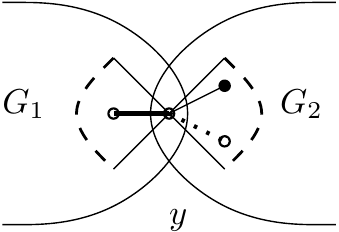}
\end{minipage} & \begin{minipage}{.20\textwidth}
\includegraphics[width=\linewidth, height=25mm]{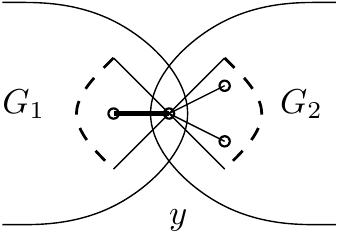}
\end{minipage} & \begin{minipage}{.20\textwidth}
\includegraphics[width=\linewidth, height=25mm]{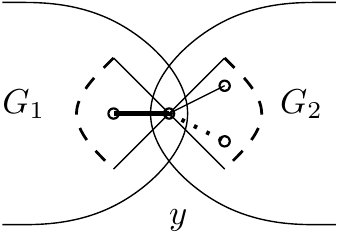}
\end{minipage}    \\
\cline{3-7}
&  & $j_{\ell}=2,i_r=2$ & $j_{\ell}=2,i_r=3$ & $j_{\ell}=2,i_r=4$ & $j_{\ell}=2,i_r=5$ & $j_{\ell}=2,i_r=6$   \\
\cline{3-7}
&  & & $j_{\ell}=3,i_r=2$ & $j_{\ell}=4,i_r=2$ & $j_{\ell}=5,i_r=2$ & $j_{\ell}=6,i_r=2$   \\
\cline{2-7}
& 3 & \begin{minipage}{.20\textwidth}\vspace{2mm}
	\includegraphics[width=\linewidth, height=25mm]{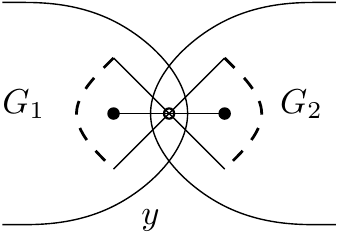}\vspace{2mm}
\end{minipage} & \begin{minipage}{.20\textwidth}
\includegraphics[width=\linewidth, height=25mm]{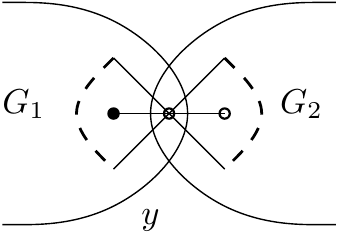}
\end{minipage} & \begin{minipage}{.20\textwidth}
\includegraphics[width=\linewidth, height=25mm]{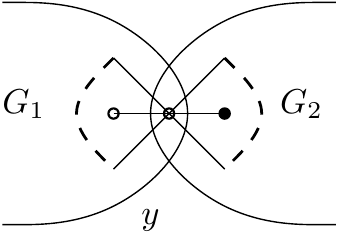}
\end{minipage} &  &    \\
\cline{3-7}
&  & $j_{\ell}=3,i_r=3$ & $j_{\ell}=3,i_r=5$ & $j_{\ell}=5,i_r=3$ & &    \\
\cline{2-7}
& 4 & \begin{minipage}{.20\textwidth}\vspace{2mm}
	\includegraphics[width=\linewidth, height=25mm]{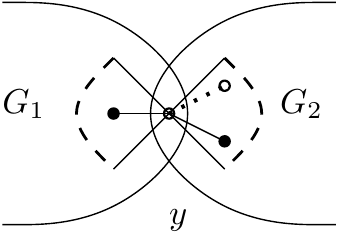}\vspace{2mm}
\end{minipage} & \begin{minipage}{.20\textwidth}
\includegraphics[width=\linewidth, height=25mm]{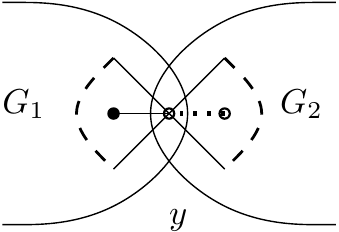}
\end{minipage} & \begin{minipage}{.20\textwidth}
\includegraphics[width=\linewidth, height=25mm]{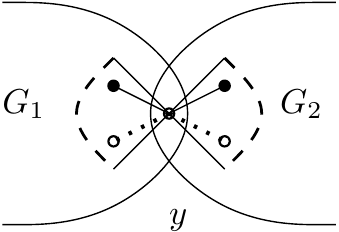}
\end{minipage} & \begin{minipage}{.20\textwidth}
\includegraphics[width=\linewidth, height=25mm]{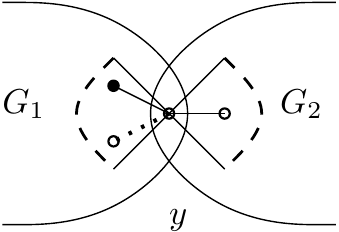}
\end{minipage} & \begin{minipage}{.20\textwidth}
\includegraphics[width=\linewidth, height=25mm]{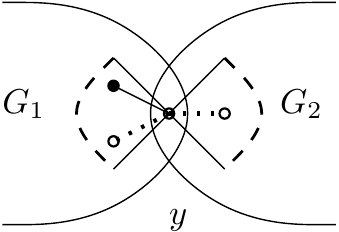}
\end{minipage}    \\
\cline{3-7}
&  & $j_{\ell}=3,i_r=4$ & $j_{\ell}=3,i_r=6$ & $j_{\ell}=4,i_r=4$ & $j_{\ell}=4,i_r=5$ & $j_{\ell}=4,i_r=6$   \\
\cline{3-7}
&  & $j_{\ell}=4,i_r=3$ & $j_{\ell}=6,i_r=3$ &  & $j_{\ell}=5,i_r=4$ & $j_{\ell}=6,i_r=4$   \\
\cline{2-7}
& 5 & \begin{minipage}{.20\textwidth}\vspace{2mm}
	\includegraphics[width=\linewidth, height=25mm]{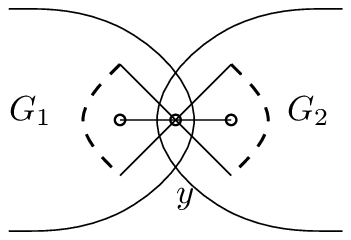}\vspace{2mm}
\end{minipage} &  &  &  &  \\
\cline{3-7}
&  & $j_{\ell}=5,i_r=5$  &  &  &  &  \\

\cline{2-7}
& 6 & \begin{minipage}{.20\textwidth}\vspace{2mm}
	\includegraphics[width=\linewidth, height=25mm]{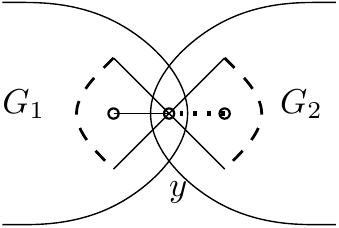}\vspace{2mm}
\end{minipage} & \begin{minipage}{.20\textwidth}
\includegraphics[width=\linewidth, height=25mm]{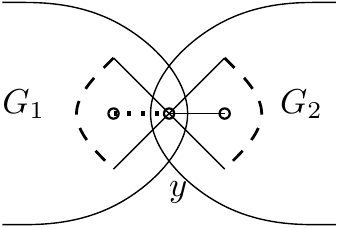}
\end{minipage} & \begin{minipage}{.20\textwidth}
\includegraphics[width=\linewidth, height=25mm]{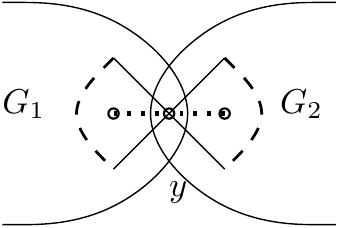}
\end{minipage} &  &    \\
\cline{3-7}
&  & $j_{\ell}=6,i_r=5$ & $j_{\ell}=5,i_r=6$ & $j_{\ell}=6,i_r=6$ &  &    \\
\cline{2-7}
\end{tabular}
}
\end{table}

\section*{Acknowledgements}

This article has been written while the fourth author was in a sabbatical visit to University of Auckland. He would like to express his gratitude to Prof.~Cristian S.~Calude and his research group for the nice and friendly hospitality. 


\end{document}